\documentclass[12pt,a4paper]{article}

\usepackage{graphicx}
\usepackage{multicol}
\usepackage{footmisc}

\usepackage{amsmath}

\usepackage{xcolor}
\usepackage{hyperref}
\usepackage{braket}
\usepackage{enumerate}
\usepackage{mathrsfs}
\usepackage{amsfonts}
\usepackage{amsthm}

\theoremstyle{plain}
\newtheorem{definition}{Definition}
\newtheorem{lemma}[definition]{Lemma}
\newtheorem{corollary}[definition]{Corollary}
\newtheorem{proposition}[definition]{Proposition}
\newtheorem{theorem}[definition]{Theorem}
\theoremstyle{definition}
\newtheorem{example}[definition]{Example}
\theoremstyle{remark}
\newtheorem{remark}[definition]{Remark}

\usepackage{tikz}
\usetikzlibrary{automata,positioning}

\title{Absorption and Fixed Points for Semigroups of Quantum Channels}
\author{Federico Girotti\\
\bigskip Dipartimento di Matematica dell'Universit\`a di Pavia\\Via Ferrata 1, 27100 Pavia, Italy\\e-mail: \texttt{f.girotti@campus.unimib.it }\\ORCID: \texttt{https://orcid.org/0000-0001-9303-1428}}

\begin{document}
\maketitle

{\noindent\bf Abstract.} In the present work we review and refine some results about fixed points of semigroups of quantum channels. Noncommutative potential theory enables us to show that the set of fixed points of a recurrent semigroup is a $W^*$-algebra; aside from the intrinsic interest of this result, it brings an improvement in the study of fixed points by means of absorption operators (a noncommutative generalization of absorption probabilities): under the assumption of absorbing recurrent space (hence allowing non-trivial transient space) we can provide a description of the fixed points set  and a probabilistic characterization of when it is a $W^*$-algebra in terms of absorption operators. Moreover we are able to exhibit an example of a recurrent semigroup which does not admit a decomposition of the Hilbert space into orthogonal minimal invariant domains (contrarily to the case of classical Markov chains and positive recurrent semigroups of quantum channels). 




\section{Introduction}
Quantum channels are mathematical objects used to described the most general evolution an open quantum system can undergo and semigroups of quantum channels (which in continuous time are known as quantum Markov semigroups) have been used for decades to model time evolution of open quantum systems under reasonable assumptions and approximations. From a mathematical point of view they represent an interesting noncommutative generalization of transition kernels and classical Markov semigroups. Many fundamental concepts and tools from the theory of classical Markov processes have been carried to the setting of semigroups of quantum channels and have provided valuable contributions to their study.

\bigskip The notion of positive recurrence for semigroups of quantum channels traces back to the 70s and it is a fundamental tool for instance for studying the long-time behaviour of quantum systems; taking inspiration from the theory of classical Markov semigroups, transience and the distinction between positive and null recurrence were further analyzed in \cite{FR2003,Uma} (in finite dimensional quantum systems, as in the case of Markov chains on a finite state space, null recurrence does not show up and the situation is far less complicated). However, there are still some open issues in the theory of recurrence for semigroups of quantum channels; following on from \cite{CG}, in this work we present some recent results that improve the understanding of null recurrence in the noncommutative setting.

Another topic which has drown attention in the last years is the study of the fixed points of the evolution: they are relevant for the asymptotics of the evolution of quantum systems (\cite{FV}) and whether or not they are a $W^*$-algebra has implications, for instance, on the relationship between conserved quantities and symmetries of the semigroup (see \cite{AGG} about the general problem of when fixed points are an algebra and \cite{Gh} for a discussion about Noether-type results in the context of quantum channels). Fixed points are well understood in the case of positive recurrent semigroups (\cite{CJ,FSU19,JP}), while less is known for general semigroups (\cite{Albert,CG,GK}). In the present work we review and improve some results of \cite{CG} which, under mild assumptions, characterize the fixed points sets in terms of absorption operators, which are a noncommutative generalization of absorption probabilities introduced always in \cite{CG}.

\bigskip Let us briefly present the setting. Let $\mathfrak{h}$ be a separable Hilbert space; we denote by $L^1(\mathfrak{h})$ trace class operators and by $B(\mathfrak{h})$ bounded linear operators; positive trace class operators with unit trace are called states and play the role of noncommutative probability densities. We recall that the topological dual of $L^1(\mathfrak{h})$ is isometric to $B(\mathfrak{h})$ via the following correspondence:
\begin{equation} \label{eq:duality}
B(\mathfrak{h}) \ni x \mapsto {\rm tr}(x\cdot) \in {L^1(\mathfrak{h})}^*.
\end{equation}
A quantum channel $\Phi:B(\mathfrak{h}) \rightarrow B(\mathfrak{h})$ is a completely positive unital $w^*$-continuous linear map, while the predual map $\Phi_*$ (defined via equation \eqref{eq:duality}) is a completely positive trace preserving map acting on $L^1(\mathfrak{h})$. A semigroup of quantum channels is a collection of quantum channels $(\mathcal{P}_t)_{t \in \mathfrak{T}}$ indexed by a semigroup $\mathfrak{T}$ such that
\[
\mathcal{P}_0={\rm Id}, \qquad \mathcal{P}_s \mathcal{P}_t=\mathcal{P}_{t+s}  \text{ for any }s,t \in \frak{T}.
\]
The collection of predual maps $(\mathcal{P}_{*t})$ forms a semigroup too.
We will mainly consider two cases: ${\frak T}=\mathbb{N}$ and the semigroup consists of the powers of a single quantum channel $\mathcal{P}:=(\Phi^n)_{n \in \mathbb{N}}$, and ${\frak T}=[0,+\infty)$ and we ask for the map $t \mapsto \mathcal{P}_t$ to be $w^*$-pointwise continuous (these semigroups are known in the literature as quantum Markov semigroups).

\bigskip In Section \ref{sec:prel} we briefly recall some fundamental concepts about reducibility of quantum channels. After that, we recap the basic notions of recurrence and transience for semigroups of quantum channels.

In Section \ref{sec:fp} we show that the fixed points set of a recurrent semigroup of quantum channels is a $W^*$-algebra. Moreover, assuming that the recurrent space is absorbing, we provide a description of the fixed points set of the semigroup in terms of absorption operators and we present some relevant consequences (these results have already been proved in \cite{CG} under stronger assumptions). We recall that the set of fixed points (sometimes called harmonic operators) of a semigroup of quantum channels is the following set:
\[\mathcal{F}(\mathcal{P}):=\{x \in B(\mathfrak{h}): \mathcal{P}_t(x)=x, \forall t \in {\frak T}\}.
\]

It is well known that in the case of classical Markov chains, recurrent states can be partitioned in communication classes; however, in Section \ref{sec:red} we show that recurrent semigroups in general do not admit a decomposition of $\mathfrak{h}$ into orthogonal minimal invariant spaces (hence the result \cite[Proposition 7.1]{CP3} for positive recurrent semigroups does not extend to recurrent ones).

\bigskip For a more detailed treatment of the topics and the results presented in this work, we refer to \cite{Gi22}.


\section{Preliminaries on semigroups of quantum channels} \label{sec:prel}

\paragraph{Reducibility.} The concept of reducibility for quantum channels was introduced in \cite{Dav} and since then have been intensively studied and exploited. We recall that for every positive operator $x$, ${\rm supp}(x):= \ker(x)^\perp$.
\begin{definition}[Enclosure] \label{def:enc}
A closed subspace ${\cal V}$ of $\mathfrak{h}$ is an enclosure (sometimes also called invariant domain) for a quantum channel $\Phi$ if, for any state $\rho$,
$$
{\rm supp}(\rho)\subseteq \mathcal{V} 
\qquad\mbox{implies}\qquad
{\rm supp}(\Phi_*(\rho))\subseteq \mathcal{V}.
$$ 
$\mathcal{V}$ is an enclosure for a semigroup $\mathcal{P}$ when it is an enclosure for any channel of the semigroup $\mathcal{P}$.\\
An enclosure $\mathcal{V}$ is said to be minimal if the only enclosures contained in $\mathcal{V}$ are the trivial ones, i.e. $\{0\}$ and $\mathcal{V}$.
\end{definition}
The following are equivalent (see \cite[Section 3]{CP1}):
\begin{itemize}
\item $\mathcal{V}$ is an enclosure for $\Phi$,
\item $p_\mathcal{V}$ is a reducing or subharmonic projection for $\Phi$, i.e. $\Phi(p_\mathcal{V})\ge p_\mathcal{V}$
\item $p_\mathcal{V} L^1(\mathfrak{h}) p_\mathcal{V}$ is hereditary for $\Phi_*$, i.e. 
$\Phi_*(p_\mathcal{V} L^1(\mathfrak{h})p_\mathcal{V})\subseteq  p_\mathcal{V} L^1(\mathfrak{h}) p_\mathcal{V}$.
\end{itemize}
When $\mathcal{V}$ is an enclosure, we can define the restricted quantum channel $\Phi^\mathcal{V}$ and the corresponding predual map:
\begin{equation}\label{eq:restriction}
\begin{split}
 \Phi^{\mathcal{V}}(p_\mathcal{V} xp_\mathcal{V}) &:= p_\mathcal{V}\Phi(p_\mathcal{V}  xp_\mathcal{V}) p_\mathcal{V} = p_\mathcal{V}\Phi(x)p_\mathcal{V}
\qquad \forall \, x\in B(\mathfrak{h}),\\
\Phi_*^{\mathcal{V}}(p_\mathcal{V} xp_\mathcal{V})&: = p_\mathcal{V}\Phi_*(p_\mathcal{V} xp_\mathcal{V}) p_\mathcal{V} =\Phi_*( p_\mathcal{V} xp_\mathcal{V})
\qquad \forall \, x\in L^1(\mathfrak{h}).
\end{split}
\end{equation}
Thank to equation \eqref{eq:restriction}, the restrictions of a semigroup of quantum channels $\mathcal{P}^\mathcal{V}$ still mantain the semigroup property. Where it does not create confusion, we will often identify $B(\mathcal{V})$ with $p_\mathcal{V} B(\mathfrak{h})p_\mathcal{V}$.

\paragraph{Recurrence and transience.} Recurrence and transience are strictly connected to reducibility properties and, of course, the ergodic behaviour of the semigroup; for the convenience of the reader, we briefly recall some related definitions and relevant properties following \cite{FR2003,Uma} (see \cite{GK} for the discrete time case). In analogy with the theory of classical Markov chains, the following space is called the positive recurrent space:
\begin{equation}
\mathcal{R}_+:=\sup\{{\rm supp}(\rho): \rho \text{ is an invariant state, i.e. $\forall \, t \geq 0$, $\mathcal{P}_{t*}(\rho)=\rho$}\}.
\end{equation}
It is well known that $\mathcal{R}_+$ is an enclosure (\cite[Proposition 3]{Uma}). For every positive $x \in B(\mathfrak{h})$, we call {\it form-potential} of $x$ the following positive symmetric closed quadratic form: $\forall\, v \in \mathfrak{h}$
$$
{\mathfrak U}(x)[v]
:= \int_0^\infty {\rm tr}(\mathcal{P}_{*t}(\ket{v}\bra{v}) x) dm(t),
\quad
D({\mathfrak U}(x)) = \left \{  v\in \mathfrak{h}\,:\, {\mathfrak U}(x)[v]<+\infty \right \}.
$$ 
$m$ stays for the counting measure in the discrete time case and for the Lebesgue measure in the continuous time setting. Loosely speaking, we can say that for any orthogonal projection $p$,  ${\frak U}(p)[v]$ represents the average time spent in ${\rm supp}(p)$ by a quantum systems that starts in the state $\ket{v}\bra{v}$ and evolves according to $\mathcal{P}$; if ${\frak U}(p)$ is bounded, no matter what is the initial state $\ket{v}\bra{v}$, the system spends on average a finite amount of time in ${\rm supp}(p)$. We define $\ker({\frak U}(x))=\{v \in \mathfrak{h}: {\frak U}(x)[v]=0\}$ and ${\rm supp}({\frak U}(x))=\ker({\frak U}(x))^\perp$; the transient space is given by 
\[\mathcal{T}=\sup\{{\rm supp}({\frak U}(x)): x \in B(\mathfrak{h}), x\geq 0, {\frak U}(x) \text{ bounded}\}.
\]
$\mathcal{R}:=\mathcal{T}^\perp$ is an enclosure and $\mathcal{T} \perp \mathcal{R}_+$ (see \cite[Proposition 7 and Corollary 2]{Uma}), hence it is meaningful to call $\mathcal{R}$ the recurrent space and to define $\mathcal{R}_0:=\mathcal{R} \cap \mathcal{R}_+^\perp$ the null recurrent space; as in the classical case, it turns out that the null recurrent space is an enclosure too (\cite[Theorem 9]{CG}). We say that the semigroup $\mathcal{P}$ is positive recurrent (resp. null recurrent/recurrent/transient) if $\mathfrak{h}=\mathcal{R}_+$ (resp. $\mathcal{R}_0$, $\mathcal{R}$, $\mathcal{T}$); in general, a semigroup $\mathcal{P}$ does not need to be of one of the above types, but \cite[Theorem 9]{Uma} shows that it can always be decomposed in its transient and recurrent restrictions, $\mathcal{P}^{\mathcal{T}}$ and $\mathcal{P}^{\mathcal{R}}$, and the recurrent restriction splits again into the positive and null recurrent restrictions, $\mathcal{P}^{\mathcal{R}_+}$ and $\mathcal{P}^{\mathcal{R}_0}$ (since $\mathcal{R}$, $\mathcal{R}_+$ and $\mathcal{R}_0$ are enclosures, the corresponding restrictions are defined as in equation \eqref{eq:restriction}).


\section{Absorption operators to describe fixed points} \label{sec:fp}

\paragraph{Recurrent semigroups.} In the first part of this section, unless differently specified, we assume that the semigroup $\mathcal{P}$ is recurrent, i.e. $\mathfrak{h}=\mathcal{R}$; notice that, in the case of a generic semigroup $\mathcal{P}$, everything that we prove holds true for the recurrent restriction $\mathcal{P}^{\mathcal{R}}$. In the case of recurrent semigroups, $\mathcal{F}(\mathcal{P})$ is tightly related to communication properties of the semigroup: every projection in $\mathcal{F}(\mathcal{P})$ is by definition the projection onto an enclosure and the vice versa is also true.
\begin{theorem}[Theorem 9, \cite{CG}] \label{thm:deco}
If ${\cal V},{\cal Z}$ are increasing enclosures included in $\mathcal{R}$, i.e. such that ${\cal V}\subseteq {\cal Z} \subseteq \mathcal{R}$, then ${\cal Z}\cap{\cal V}^\perp$ is an enclosure.
\end{theorem}
Hence projections corresponding to enclosures are harmonic projections. In general $\mathcal{F}(\mathcal{P})$ is only a selfadjoint $w^*$-closed linear space, however the following result proves that the fixed points set of a recurrent semigroup is a $W^*$-algebra, hence it is completely determined by its projections.

\begin{proposition} \label{prop:recfp}
If $\mathcal{P}$ is recurrent, then $\mathcal{F}(\mathcal{P})$ is a $W^*$-algebra.
\end{proposition}
We remark that this fact was already known in case of positive recurrent semigroups, i.e. when $\mathcal{R}=\mathcal{R}_+$ (see \cite[Theorem 2.3]{AGG}).
\begin{proof}
If $\mathcal{F}(\mathcal{P})$ is not an algebra, \cite[Lemma 2.2]{AGG} shows that there exists $x \in \mathcal{F}(\mathcal{P})$ such that $x^*x \not\in \mathcal{F}(\mathcal{P})$. Kadison-Schwarz inequality implies that
\begin{equation} \label{eq:incrnet}
\mathcal{P}_t(x^*x) \geq \mathcal{P}_t(x^*)\mathcal{P}_t(x)=x^*x \text{ for all } t \geq 0,
\end{equation}
which means that $x^*x$ is subharmonic and $(\mathcal{P}_t(x^*x))$ is a bounded positive monotone increasing net; we call $y$ its least upper bound. $y$ is a fixed point, since $\forall s \in {\frak T}$
\[\mathcal{P}_s(y)=\lim_{t \rightarrow +\infty}\mathcal{P}_s(\mathcal{P}_t (x^*x))=\lim_{t \rightarrow +\infty}\mathcal{P}_{t+s} (x^*x)=y.
\]
Notice that
\begin{itemize}
\item $\forall t \in {\frak T}$, $\mathcal{P}_t(y-x^*x)=y-\mathcal{P}_t(x^*x) \leq y-x^*x$ (superharmonic);
\item $\mathcal{P}_t(y-x^*x) \searrow 0$.
\end{itemize}
Moreover $y-x^*x \neq 0$: since $x^*x \not\in \mathcal{F}(\mathcal{P})$, there must be a time $\bar{t}>0$ such that the inequality in equation \eqref{eq:incrnet} is not an equality, hence
\[y-x^*x\geq \mathcal{P}_{\bar{t}}(x^*x) -x^*x\geq 0 \text{ and } \mathcal{P}_{\bar{t}}(x^*x) -x^*x\neq 0.
\]
By \cite[Theorem 4]{FR2003}, there exists a non-null positive operator $z$ such that ${\frak U}(z)=y-x^*x$ is bounded, hence $\{0\} \neq {\rm supp}({\frak U}(z)) \subset \mathcal{T}$ and the transient space is non-trivial.
\end{proof}
\begin{corollary} \label{coro:proj}
For every enclosure $\mathcal{V}$, $p_\mathcal{V} \in \mathcal{F}(\mathcal{P})$ and 
\[\mathcal{F}(\mathcal{P})=\overline{{\rm span}\{p_\mathcal{V}: \mathcal{V} \text{ enclosure}\}}^{\| \text{ }\|_\infty}.
\]
\end{corollary}
\begin{proof}
As we already pointed out, Theorem \ref{thm:deco} shows that projections corresponding to enclosures are harmonic projections and, since $\mathcal{F}(\mathcal{P})$ is a $W^*$-algebra, it is the norm-closure of the linear span of its projections.
\end{proof}
Another immediate consequence is a nice diagonal structure of the elements in $\mathcal{F}(\mathcal{P})$.
\begin{corollary} \label{coro:diag}
Let $x \in \mathcal{F}(\mathcal{P})$, then $x=p_{\mathcal{R}_+}x p_{\mathcal{R}_+} + p_{\mathcal{R}_0}x p_{\mathcal{R}_0}$.
\end{corollary}
\begin{proof}
By Corollary \ref{coro:proj}, it is enough to prove the statement for the projections corresponding to the enclosures, hence the result follows from \cite[Corollary 17]{CG}.
\end{proof}

\paragraph{Absorbing recurrent space.} We now turn our attention to the fixed points set of semigroups with non-trivial transient part. For every enclosure $\mathcal{V}$, one can consider the corresponding absorption operator, defined as
\[A(\mathcal{V}):= w^*-\lim_{t \rightarrow +\infty} \mathcal{P}_t(p_\mathcal{V}).
\]
Absorption operators are a noncommutative generalization of absorption probabilities: for every $v \in \mathfrak{h}$
\[\langle v, A(\mathcal{V}) v \rangle :=\lim_{t \rightarrow +\infty} {\rm tr}(\mathcal{P}_{*t}(\ket{v}\bra{v})p_\mathcal{V})
\]
can be interpreted as the asymptotic probability of finding in $\mathcal{V}$ a quantum systems that starts in the state $\ket{v}\bra{v}$ and evolves according to $\mathcal{P}$. Absorption operators have a convenient block structure which is related to recurrence: Theorem 14 in \cite{CG} shows that
\begin{equation} \label{eq:blocks}
A(\mathcal{V})=p_\mathcal{V} + p_{\mathcal{V}^\perp \cap \mathcal{T}} A(\mathcal{V}) p_{\mathcal{V}^\perp \cap \mathcal{T}}.
\end{equation}
Moreover absorption operators are fixed points of $\mathcal{P}$ (\cite[Proposition 4]{CG}) and, since the set of fixed points is norm closed, we know that
\[\overline{{\rm span}\{A(\mathcal{V}): \text{ $\mathcal{V}$ enclosure}\}}^{\| \text{ }\|_\infty} \subset \mathcal{F}(\mathcal{P}).\]
A natural question is to find under which conditions the reverse inclusion is true and so fixed points are completely described by absorption operators. We think it is an interesting problem because it creates a bridge between the fixed points set and the absorption dynamic and communication properties of the semigroup.

Although we still are not able to characterize the situation in full generality, we can improve \cite[Theorem 22]{CG} and show that fixed points are completely characterized by absorption operators when the recurrent space is absorbing, that is $A(\mathcal{R})=\mathbf{1}$. It means that asymptotically the evolution gets absorbed in the recurrent space, which, therefore, contains all relevant information regarding asymptotic quantities; it is not a very restrictive hypothesis: it holds true if $\mathfrak{h}$ is finite dimensional and stronger restrictions where often considered in ergodic theory (\cite{FV,GK}). Furthermore, as pointed out already in \cite{FV}, there is a wide class of quantum Markov semigroups for which checking its validity reduces to an analogous problem for a classical Markov chains.
\begin{theorem} \label{thm:fp}
If $A(\mathcal{R})=\mathbf{1}$, then
\[\mathcal{F}(\mathcal{P})=\overline{{\rm span}\{A(\mathcal{V}): \text{ $\mathcal{V}\subset \mathcal{R}$ enclosure}\}}^{\|\text{ } \|_\infty}={\cal F}(\mathcal{P}^{\mathcal{R}_+}) \oplus {\cal F}(\mathcal{P}^{\mathcal{R}_0}) \oplus p_\mathcal{T} \mathcal{F}(\mathcal{P}) p_\mathcal{T}.\]
\end{theorem}
\begin{proof}[Proof (sketch)]
Together with Proposition \ref{prop:recfp}, Theorem 2.4 in \cite{AGG} shows that there exists a conditional expectation $\mathcal{E}: B(\mathfrak{h}) \rightarrow \mathcal{F}(\mathcal{P})$ (it cannot be $w^*$-continuous unless $\mathcal{R}=\mathcal{R}_+$ by \cite[Theorem 2.1]{FV}) which can be approximated in the $w^*$-pointwise topology by the net $\left (\frac{1}{t_\alpha} \int_0^{t_\alpha} \mathcal{P}_s(\cdot) dm(s)\right )_{t_\alpha}$ for a directed set of times $(t_\alpha)$. One can show that $\mathcal{E}_{|{\cal F}(\mathcal{P}^{\mathcal{R}})}$ is an isomorphism of Banach spaces as in the proof of \cite[Proposition 21]{CG}, hence, by virtue of Corollary \ref{coro:diag},
$$\mathcal{F}(\mathcal{P})=\overline{{\rm span}\{\mathcal{E}(p_\mathcal{V}): \text{ $\mathcal{V} \subset \mathcal{R}$ enclosure}\}}^{\|\text{ } \|_\infty}.$$
Notice that, by definition of absorption operator, for every enclosure $\mathcal{V}$
$$\mathcal{E}(p_\mathcal{V})=w^*-\lim_{t_\alpha \rightarrow +\infty} \frac{1}{t_\alpha} \int_0^{t_\alpha} \mathcal{P}_s(p_\mathcal{V}) dm(s)=A(\mathcal{V}).$$
In order to prove the diagonal structure of $\mathcal{F}(\mathcal{P})$ with respect to $\mathcal{R}_+$ and $\mathcal{R}_0$, it suffices to show that it holds for absorption operators corresponding to recurrent enclosures and this follows from the fact that every enclosure $\mathcal{V} \subset \mathcal{R}$ is of the form $\mathcal{V}=(\mathcal{V} \cap \mathcal{R}_+) \oplus (\mathcal{V} \cap \mathcal{R}_0)$ (\cite[Corollary 17]{CG}) and from equation \eqref{eq:blocks}. Equation \eqref{eq:restriction} implies that for every $x \in \mathcal{F}(\mathcal{P})$, $p_{\mathcal{R}_+}xp_{\mathcal{R}_+}$ and $p_{\mathcal{R}_0}xp_{\mathcal{R}_0}$ are fixed points of the corresponding restrictions.
\end{proof}
Notice that in the case of recurrent semigroups, by equation \eqref{eq:blocks}, absorption operators coincide with harmonic projections and Theorem \ref{thm:fp} becomes Corollaries \ref{coro:proj} and \ref{coro:diag}. There are some interesting consequences of Theorem \ref{thm:fp} (hence we always assume $A(\mathcal{R})=\mathbf{1}$), whose proofs we skip because they are very similar to the ones of \cite[Proposition 23 and Proposition 29]{CG}, with suitable trivial adaptations.
\begin{enumerate}
\item Let $\mathcal{V} \subset \mathcal{R}$ be an enclosure; $A(\mathcal{V})=p_\mathcal{V} + p_\mathcal{T} A(\mathcal{V}) p_T$ (see equation \eqref{eq:blocks}) and $p_\mathcal{T} A(\mathcal{V}) p_T$ is the unique $y \in B(\mathcal{T})$ that solves
\begin{equation} \label{eq:ncls}
{\cal L}( y)=-p_\mathcal{T}{\cal L}(p_\mathcal{V})p_\mathcal{T}
\end{equation}
where ${\cal L}$ is the infinitesimal generator of $\mathcal{P}$ in continuous time or is equal to $\Phi-{\rm Id}$ in discrete time.

Equation \eqref{eq:ncls} is the analogous of the characterization of absorption probabilities as solution of a linear system: let $(X_n)_{n \geq 0}$ be a Markov chain on a countable state space $E$ with transition matrix $(p_{xy})_{x,y \in E}$. In the classical setting, enclosures correspond to closed sets: a subset $C \subset E$ is said to be closed if for every $x \in C$, $n \in \mathbb{N}$, $\mathbb{P}(X_n\not\in C|X_0=x)=0$. If recurrent states are absorbing, for every close set $C \subset \mathcal{R}$, the corresponding absorption probability is the unique solution of
\begin{equation} \label{eq:linsist}
\left\{\begin{array}{l}
A(C)_x=1 \mbox{ if $x \in C$}, \\
\sum_{y \in \mathcal{T}} (p_{xy}- \delta_{xy} )A(C)_y=-\sum_{y\in C} p_{xy} \mbox{ if $x \in \mathcal{T}$}.\\
\end{array}\right.
\end{equation}
\item Every enclosure $\mathcal{V}$ is of the form
\begin{equation}\label{eq:encl}
\mathcal{V}=(\mathcal{V} \cap \mathcal{R}_+)\oplus (\mathcal{V} \cap \mathcal{R}_0)\oplus(\mathcal{V} \cap \mathcal{T})
\end{equation}
where $p_{\mathcal{V} \cap \mathcal{T}} \leq A(\mathcal{V} \cap \mathcal{R})-(\mathcal{V} \cap \mathcal{R})$ and $\mathcal{V} \cap \mathcal{R} \neq \{0\}$.

As we mentioned in the previous section, enclosures play a fundamental role in the study of semigroups of quantum channels and in applications it is useful to find them; equation \eqref{eq:encl} ensures that, at least when $A(\mathcal{R})=1$, we can restrict our search to projections which commutes with $p_{\mathcal{R}_+}$, $p_{\mathcal{R}_0}$ and $p_\mathcal{T}$. It would be extremely interesting to know whether or not a subharmonic projection always commutes with $p_{\mathcal{R}_+}$, $p_{\mathcal{R}_0}$ and $p_\mathcal{T}$.

The last part of the statement tells us that every enclosure is composed by a non-null recurrent enclosure $\mathcal{V} \cap \mathcal{R}$, plus a linear space which gets completely absorbed into $\mathcal{V} \cap \mathcal{R}$.
\item $\mathcal{F}(\mathcal{P})$ is a $W^*$-algebra if and only if for every pair of orthogonal enclosures $\mathcal{V},{\cal W} \subset \mathcal{R}$, $A(\mathcal{V}) A({\cal W})=A({\cal W})A(\mathcal{V})=0$.

This last point provides a nice probabilistic characterization of the fact that $\mathcal{F}(\mathcal{P})$ is an algebra: it happens if and only if every state $\rho$ cannot reach two orthogonal recurrent enclosures.
\end{enumerate}

\section{Reducibility of recurrent semigroups} \label{sec:red}

Whether or not a recurrent semigroup admits a decomposition of the Hilbert space $\mathfrak{h}$ into orthogonal minimal enclosures (DOME) is a natural question because it is true for classical Markov chains and it was shown in many papers (for instance, see \cite[Proposition 7.1]{CP3}) for positive recurrent semigroups. In this section we will exhibit an example of a recurrent semigroup which does not admit a DOME, but before, we need to reformulate the problem thanks to Corollary \ref{coro:proj}.
\begin{lemma}
A recurrent semigroup $\mathcal{P}$ admits a DOME if and only if $\mathcal{F}(\mathcal{P})$ is an atomic $W^*$-algebra.
\end{lemma}
We recall that a $W^*$-algebra is atomic if for every projection $p$, there exists a minimal projection $q \leq p$ and this implies that there exists a denumerable family of orthogonal minimal projections $(q_\alpha)_{\alpha \in A}$ such that $\sum_{\alpha \in A} q_\alpha =\mathbf{1}$.

We remark that any DOME of $\mathcal{R}$ must be compatible with $\mathcal{R}_+$ and $\mathcal{R}_0$.
\begin{lemma}
Let $(\mathcal{V}_\alpha)$ be a DOME of $\mathcal{R}$, then $(\mathcal{V}_\alpha \cap \mathcal{R}_+)$ and $(\mathcal{V}_\alpha \cap \mathcal{R}_0)$ are DOMEs for $\mathcal{R}_+$ and $\mathcal{R}_0$, respectively.
\end{lemma}
\begin{proof}
Corollary \ref{coro:diag} implies that every (sub)harmonic projection $p_\mathcal{V}$ for $\mathcal{P}^{\mathcal{R}}$ commutes with $p_{\mathcal{R}_+}$, $p_{\mathcal{R}_0}$ and, since $\mathcal{R}_+$ and $\mathcal{R}_0$ are enclosures and the intersection of two enclosures is again an enclosure, $\mathcal{V} \cap \mathcal{R}_+$ and $\mathcal{V} \cap \mathcal{R}_0$ are enclosures. Therefore a DOME of $\mathcal{R}$ $(\mathcal{V}_\alpha)$ induces DOMEs $(\mathcal{V}_\alpha \cap \mathcal{R}_+)$ and $(\mathcal{V}_\alpha \cap \mathcal{R}_0)$ for $\mathcal{R}_+$ and $\mathcal{R}_0$, respectively (since $\mathcal{V}_\alpha$ is minimal, $\mathcal{V}_\alpha \cap \mathcal{R}_+$ and $\mathcal{V}_\alpha \cap \mathcal{R}_0$ are either $\{0\}$ or $\mathcal{V}_\alpha$).
\end{proof}
We already know that there always exists a DOME of $\mathcal{R}_+$, hence the existence of a DOME of $\mathcal{R}$ and $\mathcal{R}_0$ are equivalent problems.

\begin{example}[Noncommutative symmetric random walk on $\mathbb{Z}$] \label{ex:ncsymm}
Let us consider the group $G$ with two generators $a,b$ satisfying $a^2=b^2=e$ ($e$ is the identity element) and the corresponding left and right representations defined on $\mathfrak{h}=\ell^2(G)$. For every $g \in G$, $\lambda(g)$ and $\rho(g)$ are the unitary operators acting in the following way on the canonical basis $\mathscr{C}=\{\delta_g:g \in G\}$:
\[\lambda(g) \delta_h=\delta_{gh}, \quad \rho(g) \delta_h=\delta_{hg^{-1}} \quad \forall g, h \in G.
\]
Notice that $\lambda(g)^*=\lambda(g^{-1})$ and $\rho(g)^*=\rho(g^{-1})$ for every $g \in G$.
We introduce the following notation:
\[L(G)=\{\lambda(g):g \in G\}^{\prime\prime}=\{\lambda(a), \lambda(b)\}^{\prime\prime},\quad R(G)=\{\rho(g):g \in G\}^{\prime\prime}=\{\rho(a), \rho(b)\}^{\prime\prime}.
\]
We recall that $R(G)=L(G)^\prime$ ( \cite[Section V.7]{Ta}).
We define the quantum channel
\[\Phi(x)=\frac{1}{2}\lambda(a)x\lambda(a)+\frac{1}{2}\lambda(b)x\lambda(b), \quad x \in B(\mathfrak{h})
\]
and we consider the semigroup $\mathcal{P}:=(\Phi^n)_{n \in \mathbb{N}}$.
\paragraph{Invariant commutative subalgebra and null recurrence.} Consider the commutative W*-subalgebra of operators which are diagonal in the canonical basis $\mathscr{C}$ and its predual:
\begin{align*}
\Delta=\{x \in B(\mathfrak{h}): x=\sum_{g \in G} x_g\ket{\delta_g}\bra{\delta_g} \}\simeq \ell^\infty(G),\\
\Delta_*=\{x \in L^1(\mathfrak{h}): x=\sum_{g \in G} x_g\ket{\delta_g}\bra{\delta_g}\} \simeq \ell^1(G).
\end{align*}
Since $G$ is countable, we can relable its elements with integer numbers:
\[ \dots aba \mapsto -3 \quad ba \mapsto -2 \quad  a \mapsto -1 \quad e \mapsto 0 \quad b \mapsto 1 \quad ab \mapsto 2 \quad bab \mapsto 3\dots
\]
and this provides isomorphisms between $\ell^\alpha(G)$ and $\ell^\alpha(\mathbb{Z})$ for $\alpha \in \{1,2,\infty\}$. $\lambda(a), \lambda(b)$ act on $\mathscr{C}$ in the following way (for any $g \in G$, the label $g$ stays for $\delta_g$):
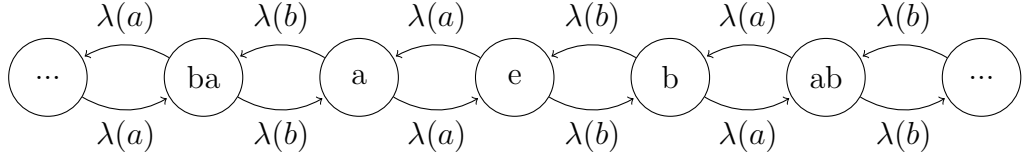
\begin{center}
    \begin{tikzpicture}
        
        \node[state]             (e) {e};
        \node[state, right=of e] (b) {b};
        \node[state, right=of b] (ab) {ab};
         \node[state, right=of ab] (end) {...};
         \node[state, left=of e] (a) {a};
         \node[state, left=of a] (ba) {ba};
         \node[state, left=of ba] (start) {...};

        \draw[every loop]
            (start) edge[bend right, auto=right] node {$\lambda(a)$} (ba)
            (ba) edge[bend right, auto=right] node {$\lambda(a)$} (start)
            (ba) edge[bend right, auto=right] node {$\lambda(b)$} (a)
            (a) edge[bend right, auto=right] node {$\lambda(b)$} (ba)
            (a) edge[bend right, auto=right] node {$\lambda(a)$} (e)
            (e) edge[bend right, auto=right] node {$\lambda(a)$} (a)
            (e) edge[bend right, auto=right] node {$\lambda(b)$} (b)
            (b) edge[bend right, auto=right] node {$\lambda(b)$} (e)
            (b) edge[bend right, auto=right] node {$\lambda(a)$} (ab)
            (ab) edge[bend right, auto=right] node {$\lambda(a)$} (b)
            (ab) edge[bend right, auto=right] node {$\lambda(b)$} (end)
            (end) edge[bend right, auto=right] node {$\lambda(b)$} (ab);
    \end{tikzpicture}
\end{center}
    
It is easy to see that $\Phi$ preserves $\Delta$ and $\Delta_*$ and its restriction corresponds via the isomorphisms above to the transition matrix of a symmetric random walk on $\mathbb{Z}$: $\Phi(\ket{\delta_g}\bra{\delta_g})=\frac{1}{2}(\ket{\delta_{ag}}\bra{\delta_{ag}}+\ket{\delta_{bg}}\bra{\delta_{bg}})$. The symmetric random walk on $\mathbb{Z}$ is null recurrent and this implies that also $\mathcal{P}$ is null recurrent.
\begin{proposition} \label{prop:nr}
$\mathcal{P}$ is null recurrent.
\end{proposition}
\begin{proof}
1. Consider a non-null positive operator $x$, then there must exists some $g \in G$ such that $\langle \delta_g,x\delta_g\rangle \geq c>0$ for some positive constant $c$. By the symmetry of the semigroup, we can assume $g=e$.
\[{\frak U}(x)[ \delta_e]=\sum_{k=0}^{+\infty} {\rm tr}(\Phi^k_*(\ket{\delta_e}\bra{\delta_e})x)\geq \sum_{k=0}^{+\infty} p^k_{0,0} \langle \delta_e,x\delta_e\rangle\geq c \sum_{k=0}^{+\infty} p^k_{0,0}=+\infty,
\]
where $p^k_{0,0}$ is the probability that a symmetric random walk on $\mathbb{Z}$ that starts in $0$ comes back to $0$ in $k$ steps. $\mathfrak{U}(x)$ is unbounded, hence $\mathcal{T}=\{0\}$.\\
2. Suppose there exists an invariant state $\rho$; the action on $\rho$ on $\Delta$ is represented by a state $\tilde{\rho}\in \Delta_*$, which must be invariant for the symmetric random walk on $\mathbb{Z}$, but this is again a contradiction, hence $\mathcal{R}_+=\{0\}$.
\end{proof}

\paragraph{$\mathcal{F}(\mathcal{P})$ has no minimal projections.} When $\mathcal{F}(\mathcal{P})$ is an algebra (this is the case by Proposition \ref{prop:recfp}) and the semigroup is generated by a single quantum channel $\Phi$, Proposition 1 in \cite{CJ} provides us a characterization of $\mathcal{F}(\mathcal{P})$ in terms of Kraus operators of $\Phi$:
\[\mathcal{F}(\mathcal{P})=\{\lambda(a),\lambda(b)\}^\prime=L(G)^\prime=R(G).
\]
We will show something stronger than the fact that $R(G)$ is not atomic: namely, we will prove that it has no minimal projections. We denote by $Z(G):=R(G) \cap L(G)$ the center of $R(G)$; notice that $(\rho(ab)+\rho(ba))/2 \in Z(G)$: it clearly belongs to $R(G)$ and it commutes with $\rho(a)$ (by symmetry it commutes with $\rho(b)$ too):
\[\rho(a) (\rho(ab)+\rho(ba))=\rho(b)+\rho(aba)=(\rho(ab)+\rho(ba))\rho(a).
\]
Let us focus on the action of $\rho(ab)$ and $\rho(ba)=\rho(ab)^*$ on $\mathscr{C}$:
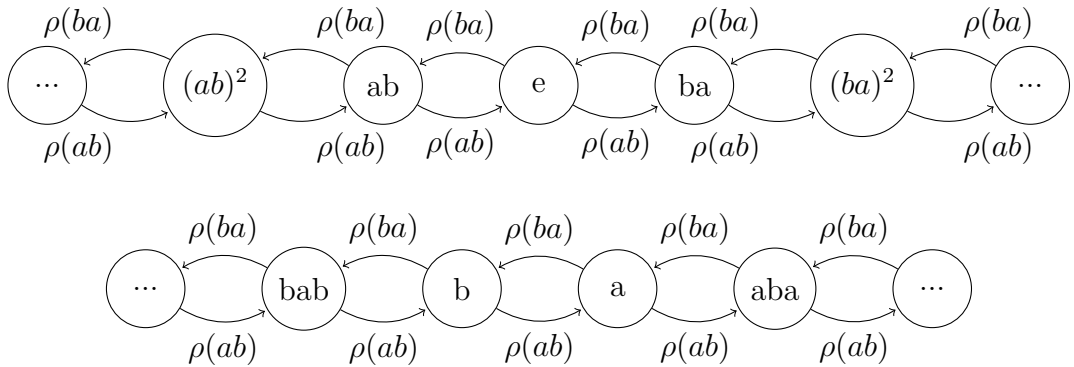
\begin{center}
    \begin{tikzpicture}
        
        \node[state]             (e) {e};
        \node[state, right=of e] (ba) {ba};
        \node[state, right=of ba] (baba) {$(ba)^2$};
         \node[state, right=of baba] (end) {...};
         \node[state, left=of e] (ab) {ab};
         \node[state, left=of ab] (abab) {$(ab)^2$};
         \node[state, left=of abab] (start) {...};

        \draw[every loop]
            (start) edge[bend right, auto=right] node {$\rho(ab)$} (abab)
            (abab) edge[bend right, auto=right] node {$\rho(ba)$} (start)
            (abab) edge[bend right, auto=right] node {$\rho(ab)$} (ab)
            (ab) edge[bend right, auto=right] node {$\rho(ba)$} (abab)
            (ab) edge[bend right, auto=right] node {$\rho(ab)$} (e)
            (e) edge[bend right, auto=right] node {$\rho(ba)$} (ab)
            (e) edge[bend right, auto=right] node {$\rho(ab)$} (ba)
            (ba) edge[bend right, auto=right] node {$\rho(ba)$} (e)
            (ba) edge[bend right, auto=right] node {$\rho(ab)$} (baba)
            (baba) edge[bend right, auto=right] node {$\rho(ba)$} (ba)
            (baba) edge[bend right, auto=right] node {$\rho(ab)$} (end)
            (end) edge[bend right, auto=right] node {$\rho(ba)$} (baba);
    \end{tikzpicture}
\end{center}
\begin{center}
    \begin{tikzpicture}
        
        \node[state]             (e) {a};
        \node[state, right=of e] (ba) {aba};
        \node[state, right=of ba] (baba) {...};
         \node[state, left=of e] (ab) {b};
         \node[state, left=of ab] (abab) {bab};
         \node[state, left=of abab] (start) {...};

        \draw[every loop]
            (start) edge[bend right, auto=right] node {$\rho(ab)$} (abab)
            (abab) edge[bend right, auto=right] node {$\rho(ba)$} (start)
            (abab) edge[bend right, auto=right] node {$\rho(ab)$} (ab)
            (ab) edge[bend right, auto=right] node {$\rho(ba)$} (abab)
            (ab) edge[bend right, auto=right] node {$\rho(ab)$} (e)
            (e) edge[bend right, auto=right] node {$\rho(ba)$} (ab)
            (e) edge[bend right, auto=right] node {$\rho(ab)$} (ba)
            (ba) edge[bend right, auto=right] node {$\rho(ba)$} (e)
            (ba) edge[bend right, auto=right] node {$\rho(ab)$} (baba)
            (baba) edge[bend right, auto=right] node {$\rho(ba)$} (ba);
           
    \end{tikzpicture}
\end{center}
Hence there exists a unitary operator $U: \ell^2(\mathbb{Z}) \otimes \mathbb{C}^2 \rightarrow \ell^2(G)$ such that $U^*\rho(ab)U$ is $S \otimes \mathbf{1}_{\mathbb{C}^2}$, where $S$ is the right shift operator. Consider the Fourier transform between the one dimensional torus $\mathbb{T}$ and $\mathbb{Z}$:
\[\begin{split}
\mathscr{F}: L^2(\mathbb{T}) &\rightarrow \ell^2(\mathbb{Z})\\
f(x) &\mapsto \hat{f}(k):= \frac{1}{2\pi}\int_\mathbb{T}f(x) e^{-ikx}dx.
\end{split}
\]
By the Fourier transform properties, it is easy to see that $\mathscr{F}^{-1} S\mathscr{F}$ is the multiplication operator $M_{e^{ix}}$ corresponding to the function $e^{ix}$. We have the following chain of equalities
\[\begin{split}
M_{\cos(x)}\otimes \mathbf{1}_{\mathbb{C}^2}&=\left (\frac{M_{e^{ix}}+M_{e^{-ix}}}{2} \right) \otimes \mathbf{1}_{\mathbb{C}^2}=\mathscr{F}^{-1} (S+S^*)\mathscr{F} \otimes \mathbf{1}_{\mathbb{C}^2}=\\
&=(\mathscr{F}^{-1} \otimes \mathbf{1}_{\mathbb{C}^2}) U^* \left (\frac{\rho(ab)+\rho(ba)}{2}\right ) U (\mathscr{F} \otimes \mathbf{1}_{\mathbb{C}^2}).\\
\end{split}
\]
Therefore the W*-algebras generated by $M_{\cos(x)}$ and $(\rho(ab)+\rho(ba))/2$ are isomorphic.
\begin{proposition}
$R(G)$ has no minimal projection.
\end{proposition}
\begin{proof}
Let us consider the following sequence of projections in the $W^*$-algebra generated by $M_{\cos(x)}$: for any $n \in \mathbb{N}$, $j=0,\dots,2^{n+1}-1$ we define $q_{j,n}$ as the indicator function of the set
\[\cos^{-1}\left (\left [-1+j 2^{-n},-1+(j+1) 2^{-n}\right ]\right );\]
we define $p_{j,n}:=U(\mathscr{F} \otimes \mathbf{1}_{\mathbb{C}^2}) (q_{j,n} \otimes \mathbf{1}_{\mathbb{C}^2}) (\mathscr{F}^{-1} \otimes \mathbf{1}_{\mathbb{C}^2})U^*$, which is the spectral projection of $(\rho(ab)+\rho(ba))/2$ corresponding to the same set and it is in $Z(G)$. Notice that by construction
\begin{enumerate}
\item $\sum_{j=0}^{,2^{n+1}-1} p_{j,n}=\mathbf{1}$,
\item $p_{j,n}p_{k,n}=\delta_{jk} p_{j,n}$ for every $n \in \mathbb{N}$, $j=0,\dots, 2^{n+1}-1$,
\item for every $p_{j,n}$ with $n \geq 1$, there exists a unique $k$ such that $p_{j,n}p_{k,n-1}=p_{j,n}$.
\end{enumerate}

Let $f \in R(G)$ be a minimal projection; since $(p_{j,n}) \subset Z(G)$, for every $n \in \mathbb{N}$ and $j \in \{0,\dots, n-1\}$
\[fp_{j,n}=p_{j,n}f=p_{j,n}fp_{j,n}
\]
is a projection dominated by $p_{j,n}$ and $f$. The minimality of $f$ implies that $fp_{j,n}$ is either $0$ or $f$ itself, and, by $1.$ and $2.$, for every $n \in \mathbb{N}$, there exists a unique $j \in \{0,\dots, n-1\}$, which we call $j(n)$, such that $p_{j,n}f=f$. Hence $f \leq p_{j(n),n}$ for every $n \in \mathbb{N}$, $p_{j(n+1),n+1} \leq p_{j(n),n}$ for every $n \in \mathbb{N}$ and $p_{j(n),n} \downarrow 0$ (in the $w^*$-topology), hence $f=0$.
\end{proof}

\begin{remark}
The following is the GKLS generator of a continuous time counterpart of the same example:
\[{\cal L}(x)=\frac{1}{2}(\lambda(a)x\lambda(a)-x)+\frac{1}{2}(\lambda(b)x\lambda(b)-x).
\]
The analysis above can be carried out also in this case.
\end{remark}
\begin{remark}
The structure of the decoherence-free subalgebra ${\cal N}(\mathcal{P})$ of a semigroup and wether it is atomic or not has been investigated in \cite{CJ,FSU19,SU21}, especially in relation to enviromental decoherence. In the present example, even ${\cal N}(\mathcal{P})$ has no minimal projection. Indeed, by \cite[Proposition 3]{CJ}, we have that
\[{\cal N}(\mathcal{P})=\{\lambda(ab), \lambda(ba)\}^\prime.
\]
Consider the decomposition $\ell^2(G)=\bigoplus_{h \in S} \ell^2(Hs)\simeq \ell^2(H) \otimes \ell^2(S)$ where $H$ is the subgroup of $G$ generated by $ ab,ba$ and $S$ is a set of representatives of $G/H$. Notice that ${\cal N}(\mathcal{P})$ is $(L(H) \otimes \mathbf{1}_{\ell^2(S)})^\prime=R(H) \otimes B(\ell^2(S))$, hence its center is $L(H)\cap R(H) \otimes \mathbf{1}_{\ell^2(S)}=\{\lambda(ab), \lambda(ba)\}^{\prime \prime}\cap  \{\rho(ab),\rho(ba)\}^{\prime\prime}$. $(\rho(ab)+\rho(ba))/2$ is in the center of ${\cal N}(\mathcal{P})$ too, hence we can repeat the same proof as for $\mathcal{F}(\mathcal{P})$.
\end{remark}
\end{example}

Although Example \ref{ex:ncsymm} shows that in general ${\cal F}(\mathcal{P}^{\mathcal{R}})$ is not atomic, \cite[Theorem 2.4]{AGG} together with Proposition \ref{prop:recfp} shows that it is injective. A natural question is if it possible to find \textit{interesting} features of the semigroup ensuring the atomicity of ${\cal F}(\mathcal{P}^{\mathcal{R}})$ and, hence, the decomposition of $\mathcal{R}$ into orthogonal minimal enclosures.

\bigskip {\noindent\bf Ackonwledgements.} The author acknowledges the support of the INDAM GNAMPA project 2020 ``Evoluzioni markoviane quantistiche'' and of
the Italian Ministry of Education, University and Research (MIUR) for the Dipartimenti di Eccellenza Program (2018–
2022)—Dept. of Mathematics ``F. Casorati'', University of Pavia. Special aknowledgments go to the organizers of the \textit{$41^\text{st}$ International Conference on Quantum Probability and Related Topics} and to Raffaella Carbone for introducing the author to the problem and for many useful suggestions.

\bibliographystyle{abbrv}
\bibliography{biblio}
\end{document}